\documentclass[11pt]{article}

\usepackage[letterpaper,margin=1in]{geometry}
\usepackage[T1]{fontenc}
\usepackage[utf8]{inputenc}
\usepackage{palatino} % Palatino-like text, close to the cited SOSA 2026 paper.
\usepackage{microtype}

\usepackage{amsmath,amssymb,amsthm,mathtools}
\usepackage{aliascnt}
\usepackage{enumitem}
\usepackage{booktabs}
\usepackage{tabularx}
\usepackage{graphicx}
\usepackage{xcolor}
\usepackage{xspace}

\setlist[itemize]{topsep=0.5em,itemsep=0.35em,parsep=0pt}
\setlist[enumerate]{topsep=0.5em,itemsep=0.35em,parsep=0pt}

\usepackage[most]{tcolorbox}
\newtcolorbox{algorithmBox}[1]{
  enhanced,
  breakable,
  colback=gray!4,
  colframe=black!60,
  boxrule=0.6pt,
  arc=2pt,
  left=8pt,
  right=8pt,
  top=6pt,
  bottom=6pt,
  title={#1},
  fonttitle=\bfseries,
  coltitle=black,
  attach boxed title to top left={
    xshift=8pt,
    yshift*=-\tcboxedtitleheight/2
  },
  boxed title style={
    colback=white,
    colframe=black!60,
    boxrule=0.6pt,
    arc=2pt
  }
}
\definecolor{SOSARefRed}{RGB}{175,45,35}
\definecolor{SOSACiteGreen}{RGB}{40,120,55}
\definecolor{SOSAUrlBlue}{RGB}{35,85,145}
\usepackage[
  colorlinks=true,
  linkcolor=SOSARefRed,
  citecolor=SOSACiteGreen,
  urlcolor=SOSAUrlBlue,
  pdfborder={0 0 0}
]{hyperref}
\theoremstyle{plain}
\newtheorem{theorem}{Theorem}[section]

\newaliascnt{lemma}{theorem}
\newtheorem{lemma}[lemma]{Lemma}
\aliascntresetthe{lemma}
\newaliascnt{proposition}{theorem}

\aliascntresetthe{proposition}
\newaliascnt{corollary}{theorem}
\newtheorem{corollary}[corollary]{Corollary}
\aliascntresetthe{corollary}
\newaliascnt{claim}{theorem}

\aliascntresetthe{claim}
\newaliascnt{conjecture}{theorem}
\newtheorem{conjecture}[conjecture]{Conjecture}
\aliascntresetthe{conjecture}

\theoremstyle{definition}
\newaliascnt{definition}{theorem}
\newtheorem{definition}[definition]{Definition}
\aliascntresetthe{definition}
\newaliascnt{problem}{theorem}

\aliascntresetthe{problem}
\newaliascnt{example}{theorem}

\aliascntresetthe{example}

\theoremstyle{remark}
\newaliascnt{remark}{theorem}

\aliascntresetthe{remark}
\newaliascnt{observation}{theorem}

\aliascntresetthe{observation}

\usepackage[nameinlink,noabbrev]{cleveref}

\crefname{theorem}{theorem}{theorems}
\crefname{lemma}{lemma}{lemmas}
\crefname{proposition}{proposition}{propositions}
\crefname{corollary}{corollary}{corollaries}
\crefname{definition}{definition}{definitions}
\crefname{problem}{problem}{problems}
\crefname{figure}{Figure}{Figures}
\crefname{table}{Table}{Tables}
\crefname{section}{Section}{Sections}

\newcommand{\OPT}{\mathrm{OPT}}

\newcommand{\bigO}{\ensuremath{\mathcal{O}}\xspace}
\newcommand{\lpccc}{\text{LP}_{CCC}}
\def\Exp{\mathbb{E}}
\newcommand{\tO}{\ensuremath{\widetilde{\mathcal{O}}}\xspace}

\hypersetup{
  pdftitle={CKR Partitions and Lower Bounds for Constrained
Correlation Clustering and Variants},
  pdfauthor={Florian Adriaens}
}

\begin{document}
\pagenumbering{gobble}

% =============================================================================
% Dedicated title page
% SOSA27 explicitly asks for author information. This differs from SOSA26,
% which used lightweight double-blind reviewing at the submission stage.
% =============================================================================
\begin{titlepage}
  \centering
  \vspace*{1.25cm}

  {\LARGE CKR Partitions and Lower Bounds for Constrained Correlation Clustering and Variants\par}

  \vspace{1.25cm}

  % This 2-by-2 author grid resembles the cited SOSA 2026 paper.
  % Delete unused cells or add rows as needed. For one or two authors, keep
  % only the first row. Include enough information to identify every author.
  \begin{tabularx}{0.92\textwidth}{@{}>{\centering\arraybackslash}X@{}}
    {\large Florian Adriaens}\\
    Department of Computer Science \\
    University of Helsinki \\
    {florian.adriaens@helsinki.fi}
  \end{tabularx}

  \vspace{1.15cm}

  \begin{minipage}{0.86\textwidth}
    \small
    \begin{center}
      \textbf{Abstract}
    \end{center}
By using a simple textbook reduction from vertex cover, we show that the following three problems are all UG-hard to approximate with constant-factor smaller than two; \emph{ minimum weakness strong triadic closure}, \emph{cluster deletion} and \emph{constrained correlation clustering}.\footnote{Independent and concurrent of our work, both Cao and Xu~\cite{cao2026cluster} and Azizeddin et al.~\cite{azizeddin2026constrained} proved the same hardness results for CD and CCC, using the same reduction.}
Additionally, we analyze the well-known low-diameter decomposition by Calinescu, Karloff and Rabani \cite{calinescu2005approximation} applied to the standard LP relaxation semi-metric for constrained correlation clustering. 
As opposed to traditional pivot-based approaches, a CKR partition elegantly handles must-link and cannot-link constraints.
It guarantees a 3-approximation in expectation, which matches the approximation ratio by van Zuylen and Williamson \cite{van2009deterministic}.
We conjecture that it in fact achieves a strictly better than 3-approximation, yet this remains an open problem.
  \end{minipage}

  \vfill

  % Optional funding/support notes, visually similar to the cited paper.
  % Remove this entire block when it is not needed.
  %\begin{minipage}{0.82\textwidth}
  %  \footnotesize
  %  \rule{0.36\textwidth}{0.4pt}\par
  %  \textsuperscript{*}Supported in part by Grant A.\par
  %  \textsuperscript{\ensuremath{\dagger}}Supported in part by Grant B.\par
  % \textsuperscript{\ensuremath{\ddagger}}Supported in part by Grant C.
  %\end{minipage}
\end{titlepage}

% Start the main paper at page 1, as in the cited SOSA 2026 paper.
\pagenumbering{arabic}

\section{Introduction}
\label{sec:intro}

% The SOSA27 call specifically asks the beginning of the paper to contain:
% (i) a clear problem description;
% (ii) a survey of prior work, with a candid assessment of simplicity/elegance;
% (iii) a discussion of the paper's contributions.

In a signed graph, each edge has a positive $(+)$ or negative $(-)$ label. 
This paper is concerned with finding approximate solutions to several graph modification and clustering problems in signed graphs.
The problems and their known approximation results are discussed below.
In Section~\ref{sec:contr} we summarize our contributions.
The input of all these problems are assumed to be complete signed graphs, unless mentioned otherwise.\footnote{A complete signed graph is the same as an unsigned graph, but sometimes it is more convenient to describe problems and objects in terms of negative edges instead of non-edges. We will switch between the two notions.}

\paragraph{$(i)$ Constrained Correlation Clustering.}
Correlation Clustering (CC) is a fundamental approach for clustering signed graphs, pioneered by Bansal, Blum and Chawla~\cite{bansal2004correlation}. This unsupervised learning task has found widespread applications across a variety of different fields, see the recent survey \cite{bonchi2022correlation} for an in-depth discussion.
In short, CC seeks to partition a graph's vertices into clusters, with the goal of minimizing the total number of \emph{mistakes}. Here, a partition makes a mistake if $(a)$ the endpoints of a $+$ edge are part of different clusters, or $(b)$ the endpoints of a $-$ edge belong to the same cluster.
There is no restriction on the number of clusters a partition can have.

Constrained Correlation Clustering (CCC) is a generalization of CC \cite{van2009deterministic}.
It has two additional inputs: a set of must-link pairs (these pairs need to be intra-cluster) and a set of cannot-link pairs (these pairs need to be inter-cluster).
Similarly, the goal is to minimize total mistakes, but a partition is feasible only if it does not violate any of link constraints.
This naturally captures a semi-supervised setting, where, for example, an expert guides the output by strictly enforcing hard constraints \cite{davidson2007complexity,wagstaff2000clustering,berg2017cost,guo2024efficient}.

The known approximation lower bounds for CCC are those that carry over from CC; APX-hardness \cite{CHARIKAR2005360}, NP-hardness (under randomized reductions) for constant-factor better than $\frac{24}{23}$ \cite{cao2024understanding}, and NP-hardness for constant-factor better than $\frac{2137}{2136}$ \cite{adriaenssimple}.

For a very long time, the best known approximation ratio for CCC was 3, stemming from a pivot-based algorithm by the same authors who introduced the problem \cite{van2009deterministic}. 
Only very recently, the work of \cite{azizeddin2026constrained} showed that there exists a $\frac{16}{7}$ approximation based on rounding the canonical LP relaxation, and even slightly better if one instead uses a Sherali-Adams relaxation.

This is in sharp contrast to (unconstrained) CC, for which a $5/2$ approximation was given by \cite{ailon2008aggregating}, and in recent years has seen an impressive line of work that achieves even better-than-two approximations \cite{cohen2022correlation, cohen2023handling, cohen2024combinatorial, cao2024understanding, cao2025solving}.
One such method is based on rounding the so-called \emph{Cluster LP}, which is a linear program relaxation of CC. Despite its exponential size, the program is shown to be polynomial time solvable \cite{cao2024understanding}, or even in sublinear time \cite{cao2025solving}. In a follow up work Kalavas et al.~\cite{kalavas2025towards}[Theorem 1.1] showed that if one could solve the \emph{Constrained Cluster LP} in poly-time, then CCC admits a $\approx 1.92$ approximation.

\paragraph{$(ii)$ Cluster Deletion.}
This classic problem asks for the following:
Given an unsigned graph, find the minimum number of edges to delete such that the connected components of the remaining graph are cliques \cite{SHAMIR2004173}. 
Equivalently, the remaining graph cannot have an induced $P_3$.
Cluster Deletion (CD) is a special case of CCC, by setting all the negative edges as cannot-link pairs, without any must-link pairs.

Charikar et al.~\cite{CHARIKAR2005360} first gave a 4-approximation for CD. 
Since CD is a special case of CCC, the results of van Zuylen and Williamson also imply a 3-approximation for CD \cite{van2009deterministic}.
Subsequently, Veldt et al.~\cite{veldt2018correlation} gave a 2-approximation. All these algorithms are LP-based.
The known inapproximability bounds for CD are APX-hardness \cite{SHAMIR2004173} and NP-hardness for constant-factor better than $\frac{2137}{2136}$ \cite{adriaenssimple}.

\paragraph{$(iii)$ Minimum Weakness Strong Triadic Closure.}
This problem is in some sense an easier variant of CD. Instead of \emph{deleting} edges, the Minimum Weakness Strong Triadic Closure problem (MinSTC) only asks for an edge \emph{labeling}. Namely, find a smallest sized set of \emph{labeled} edges, such that every induced $P_3$ contains a labeled edge~\cite{sintos2014using}.
The difference between CD and MinSTC is that deleting edges might create new induced $P_3$ subgraphs, so one has to be more careful when designing algorithms for CD.
There are known graphs where the optimal MinSTC and CD values differ \cite{gruttemeier2020relation}.
This problem was first introduced by Sintos and Tsaparas~\cite{sintos2014using} as a way of inferring the strength of ties in social networks \cite{easley2010networks}, by leveraging the triadic closure property from sociology \cite{granovetter1973strength}.

Sintos and Tsaparas~\cite{sintos2014using} showed that MinSTC admits a simple 2-approximation by using any 2-approximation for vertex cover on the Gallai graph \cite{le1996gallai}. The only lower bound on approximability appears to be NP-hardness for constant-factor better than $\frac{2137}{2136}$ \cite{adriaenssimple}. Gr\"uttemeier and Komusiewicz~\cite{gruttemeier2020relation} asked if there exists a constant-factor approximation for MinSTC with factor smaller than 2.

\subsection{Contributions}
\label{sec:contr}
Our first contribution is that, by using a short and elegant reduction, we show that all the above problems (CCC, CD and MinSTC) are UG-hard to approximate with factor $2-\epsilon$ for any constant $\epsilon>0$. Independent and concurrent of our work, both Cao and Xu~\cite{cao2026cluster} and Azizeddin et al.~\cite{azizeddin2026constrained} proved the same hardness results for CD and CCC, using the same reduction.

Assuming the UGC is true, this has the following implications.
\begin{itemize}
  \item For CCC, this refutes the assumption of Kalavas et al.~\cite{kalavas2025towards} that the Constrained Cluster LP can be $(1+\epsilon)$ approximated in poly-time. This result also holds for CCC instances without any must-link constraints.
  \item For CD, this shows that the approximation ratio provided by Veldt et al.~\cite{veldt2018correlation} is tight.
  \item For MinSTC, this negatively answers a question by Gr\"uttemeier and Komusiewicz~\cite{gruttemeier2020relation}.
\end{itemize}

Secondly, we show that a simple CKR partition \cite{calinescu2005approximation} can be used to 3-approximate CCC.
This matches the original guarantee by van Zuylen and Williamson~\cite{van2009deterministic}.

CKR partitions are a standard tool for decomposing metric spaces into low diameter parts, and play a fundamental role in metric embedding theory \cite{fakcharoenphol2003improved, fakcharoenphol2004tight,mendel2007ramsey}.
Prior work has used (variants of) CKR decompositions for correlation clustering before, but only in the specific context of local correlation clustering \cite{kalhan2019correlation,jafarov2021local}.\footnote{In local correlation clustering, one seeks to minimize the $\ell_p$ norm of the vector defined by counting the mistakes incident to each vertex.}
A CKR partition works differently from a typical pivot-based approach, as a CKR partition is allowed to select centers that already have been clustered in previous iterations.

\begin{theorem}[CKR partitions for CCC]
\label{thm:main}
A simple randomized decomposition by Calinescu, Karloff and Rabani \cite{calinescu2005approximation} for partitioning finite metric spaces gives a 3-approximation for CCC.
\end{theorem}

\subsection{Other related work}
Both \cite{fischerSTAC} and \cite{veldt2026simple} studied how to approximate CCC within constant factor, but without solving the expensive canonical LP relaxation. \cite{fischerSTAC} gave a 16-approximation in $\tO(n^3)$ time, while \cite{veldt2026simple} recently improved this to a $3+\epsilon$ approximation in a similar running time. 

The CD problem is known to be efficiently solvable on cographs \cite{gao2013cluster}, split graphs \cite{bonomo2015complexity}, interval graphs \cite{konstantinidis2021cluster} and other graph classes \cite{gruttemeier2020relation,komusiewicz2012cluster, bonomo2015one,galesi2025cluster,konstantinidis2025algorithms}.
Recent work has proposed faster combinatorial approximations which scale to large graphs \cite{CDbalmaseda,veldt2022correlation}.
The vertex-deletion variant of CD, which asks to remove vertices instead of edges, is known to have a 2-approximation and this is also tight under the UGC \cite{aprile2023tight}.

MinSTC remains NP-hard on graphs with degree four \cite{konstantinidis2018strong} and split graphs \cite{konstantinidis2020maximizing}.
On the other hand, the problem is solvable in polynomial time when the graph is subcubic or a co-graph \cite{konstantinidis2018strong}, bipartite \cite{sintos2014using} or a proper interval graph \cite{konstantinidis2020maximizing}.
A complete complexity dichotomy between tractable and NP-hard cases for $H$-free graphs is given in \cite{gruttemeier2020relation}, where $H$ is a graph of at most four vertices.
MinSTC has also been studied in the context of kernelization and fixed parameter tractability \cite{gruttemeier2020relation,golovach2020parameterized,CaoImproved}, see the survey \cite{crespelle2023survey} for a detailed discussion. Following the work of \cite{sintos2014using}, there have been multiple different approaches for predicting tie-strengths in social networks based on the STC principle \cite{adriaens2020relaxing, rozenshtein2017inferring,oettershagen2026consistent,oettershagen2025inferring}.

\section{Reduction}
\label{sec:reduc}
We show the UG-hardness result for Cluster Deletion (CD).
As CD is a special case of CCC, the result also holds for CCC. The hardness proof for MinSTC follows a very similar reasoning and is omitted for brevity. Our reduction is inspired by the recent work on \emph{bad triangle transversals} \cite{adriaenssimple}, which uses a somewhat similar reduction from vertex cover, but without adding a large clique.

For a graph $G=(V,E)$ with $|V|=n$, let $\tau(G)$ denote the size of a minimum vertex cover of $G$ and let $\alpha(G)$ denote the size of a maximum independent set. We will reduce from the following classic gap result by Khot and Regev \cite{khot2008vertex}.

\begin{theorem}[\cite{khot2008vertex}]
\label{thm:khot}
For every $\delta>0$, it is UG-hard to decide between $\tau(G) \geq (1-\delta)n$ and $\tau(G) \leq (\frac{1}{2}+\delta)n$.    
\end{theorem}

Construct a graph $H$ in the following way.
$H$ consists of $\bar{G}$ (the complement of $G$) and a clique $K$ of $n^3$ vertices. Connect every node in $K$ to every node in $\bar{G}$. This finalizes the construction of $H$.
Let $\OPT_{\text{CD}}(H)$ be the optimum value of the CD problem on $H$. The following two inequalities are key.

\begin{lemma}
\label{lem:ineqs}
$\tau(G)n^3 \leq \OPT_{\text{CD}}(H) \leq \tau(G)n^3 + n^2$
\end{lemma}

\begin{proof}
For the lower bound, fix a node $v \in K$. Observe that every edge in $G$ (which corresponds to a missing edge in $\bar{G}$) forms an induced $P_3$ in $H$ centered at $v$. For a feasible CD solution $C$ of the graph $H$, note that the set $\{u \in \bar{G}: (u,v) \in C\}$ must form a valid vertex cover of $G$, because if not, then some induced $P_3$ centered at $v$ does not contain an edge from $C$.
Since $|K|=n^3$ it follows that $\tau(G)n^3 \leq \OPT_{\text{CD}}(H)$.

For the upper bound, form a partition $\mathcal{P}$ of the nodes of $H$ as follows: $(a)$ group together the nodes of $K$ with the nodes from a maximum clique in $\bar{G}$, and $(b)$ all remaining nodes of $\bar{G}$ are singletons.
The partition $\mathcal{P}$ is feasible for CD because all its elements are cliques. A maximum clique in $\bar{G}$ corresponds to a maximum independent set in $G$, and we also know that $\alpha(G)+\tau(G)=n$.
Therefore, the number of singletons is exactly $\tau(G)$.
A \emph{cut} edge is an edge of $H$ with endpoints belonging to different clusters of $\mathcal{P}$.
The number of cut edges between $\bar{G}$ and $K$ is exactly $\tau(G)n^3$. There can also be cut edges inside $\bar{G}$, of which there are at most $n^2$.
\end{proof}

From Lemma~\ref{lem:ineqs} and Theorem~\ref{thm:khot} the result follows immediately, as the CD gap ratio on $H$ becomes
\begin{equation*}
\frac{(1-\delta)n^4}{(\frac{1}{2}+\delta)n^4+n^2}.  
\end{equation*}
For any $\epsilon>0$, one can make this gap ratio larger than $2-\epsilon$ by setting $\delta$ sufficiently small and $n$ large enough (depending on $\epsilon$).

\section{CKR Partitions for CCC}
\label{sec:ckr}
Let $G=(V,E^+,E^-)$ be a given signed graph, with $E^+$ (resp. $E^-$) the positive (resp. negative) edges.
A Constrained Correlation Clustering (CCC) instance consists of G, together with a set of the must-link pairs $F$ and cannot-link pairs $H$.
The canonical LP relaxation of CCC \cite{van2009deterministic} is given by
\begin{align}
\text{min} \quad &\sum_{uv \in E^+}x_{uv} + \sum_{uv \in E^-}1-x_{uv}, &\tag{LP\(_{CCC}\)}\\
\text{s.t.}\quad &x_{uv} \leq x_{uw}+x_{vw}  & \forall u,v,w \in V, \notag\\
&x_{uv} = 0 &\forall uv \in F, \notag\\
&x_{uv} = 1 &\forall uv \in H, \notag\\
&0 \leq x_{uv} \leq 1 &\forall u,v \in V. \notag
\end{align}
If the distances are binary, then $x_{uv}=0$ indicates that $u$ and $v$ belong to the same cluster, while $x_{uv}=1$ indicates they belong to different clusters. The triangle inequality ensures this forms a valid partition. The constraints on $F$ and $H$ ensure no link constraints are violated.
We assume that there exists a feasible partition that does not violate any link constraints, otherwise $\lpccc$ is infeasible.

 Consider the following randomized algorithm for rounding $\lpccc$ into feasible clusters.
The algorithm is nothing more than applying the well-known low-diameter decomposition for general metric spaces proposed by Calinescu, Karloff and Rabani \cite{calinescu2005approximation}, to the semi-metric obtained by solving $\lpccc$.
Write $V = \{v_1, \ldots, v_n\}$ and let $B(v,r) = \{u \in V: x_{uv} \leq r \}$ be the closed ball around $v$ with radius $r$.

\begin{algorithmBox}{\textsc{CKR-Pivot} for rounding $\lpccc$}
\begin{enumerate}
  \item Solve $\lpccc$ to obtain semi-metric $\{x_{uv}\}$. 
  \item Draw a uniform random permutation $\pi$ of $\{1,\ldots,n\}$.
  \item Draw a uniform random number $R \in [0,1/2)$.
  \item Define $C_1=B(v_{\pi(1)},R)$ and for $2 \leq j \leq n$ define
  \begin{equation*}
      C_j = B(v_{\pi(j)},R) \setminus \bigcup_{i<j}C_i.
  \end{equation*}
  \item Return partition $\mathcal{P} = \{C_1,\ldots,C_n\} \setminus \{\emptyset\}$.
\end{enumerate}
\end{algorithmBox}

Before diving into the analysis, let us highlight the differences with the original 3-approximation for CCC~\cite{van2009deterministic}.
The approach of van Zuylen and Williamson~\cite{van2009deterministic} consists of first constructing an auxiliary graph $\hat{G}$, followed by executing the well-known pivot algorithm \cite{ailon2008aggregating} on $\hat{G}$. The pivots can be chosen either deterministically or uniformly at random.
The construction of the auxiliary graph $\hat{G}$ is rather cumbersome. For example, if $x_{uv}=1/2$ then deciding if $uv$ becomes a positive edge in $\hat{G}$ depends on whether $u$ and $v$ belong to the same connected component induced by the must-link pairs $F$. This tie-breaking rule is needed to ensure that no link constraints are violated when executing the pivot step.

In contrast, \textsc{CKR-Pivot} works directly on the $\lpccc$ metric and differs from \cite{van2009deterministic} in three major ways:

\begin{itemize}
    \item \textsc{CKR-Pivot} is \emph{oblivious} to the must-link or cannot-link constraints. In other words, after obtaining a feasible $\lpccc$ solution, the algorithm no longer uses $F$ or $H$.
    \item In \textsc{CKR-Pivot}, a vertex $u$ will be assigned to cluster $C_j$ iff $v_{\pi(j)}$ is the first vertex according to the permutation $\pi$ among all vertices in $B(u,R)$. Note that some clusters can be empty, indeed, they might not even contain their own center. This is very different from a usual pivot method. At at any iteration, pivot selects a new pivot (a center) only from the remaining pool of unclustered vertices. A CKR partition is allowed to select centers that already have been clustered in a previous iteration.
    \item On general signed graphs, \textsc{CKR-Pivot} gives an expected $\bigO(\log n)$ ratio for CCC (see Section~\ref{sec:gensg}). This is not guaranteed by the pivot method of \cite{van2009deterministic}.
\end{itemize}

\subsection{Expected cost and link feasibility}
Here we derive the general form of the expected cost and verify that the proposed CKR rounding does not violate any must-link or cannot-link constraints.
For $u \in V$ let $\mathcal{P}(u)$ be the unique cluster of $\mathcal{P}$ containing $x$.
Let $\text{cost}(\mathcal{P})$ denote the number of mistakes made by \textsc{CKR-Pivot}.
Conditional on $R=r$, its expectation can be written as
\begin{align}
\Exp[\text{cost}(\mathcal{P}) | R=r] &= \sum_{uv \in E^+} \text{Pr}[\mathcal{P}(u) \neq \mathcal{P}(v) | R=r ] + \sum_{uv \in E^-} \text{Pr}[\mathcal{P}(u) = \mathcal{P}(v) | R=r ] \notag \\
&= \sum_{uv \in E^+}\frac{|B(u,r) \triangle B(v,r)|}{|B(u,r) \cup B(v,r)|} + \sum_{uv \in E^-}\frac{|B(u,r) \cap B(v,r)|}{|B(u,r) \cup B(v,r)|} \label{eq:expcost}.
\end{align}
 Eq.~(\ref{eq:expcost}) follows by considering the earliest vertex in $B(u,r) \cup B(v,r)$, according to $\pi$.
If this vertex is part of $B(u,r) \cap B(v,r)$, then $\mathcal{P}(u) = \mathcal{P}(v)$. If not, then $\mathcal{P}(u) \neq \mathcal{P}(v)$.

It is easy to see that the link constraints are satisfied.
If $uv \in F$, then $x_{uv}=0$. Therefore $B(u,r) \triangle B(v,r) = \emptyset$, and thus $\text{Pr}[\mathcal{P}(u) \neq \mathcal{P}(v)]=0$. Likewise, if $uv \in H$, then $x_{uv}=1$. Therefore, because $r<1/2$ we have $B(u,r) \cap B(v,r) = \emptyset$, implying $\text{Pr}[\mathcal{P}(u) = \mathcal{P}(v)] = 0$.

\subsection{Warm up: general signed graphs}
\label{sec:gensg}
So far, we have not used any properties of \emph{complete} signed graphs. Let us start by analyzing the approximation ratio in general signed graphs. 
In this case, a ratio of $\bigO(\log n)$ in expectation can be immediately retrieved, which matches the best deterministic ratio based on region growing techniques \cite{demaine2006correlation, CHARIKAR2005360}.
This is fact a folklore result, as it is well-known that a CKR decomposition gives a $\bigO(\log k)$ approximation for Multicut  (by taking a random permutation of source vertices as CKR centers \cite{williamson2011design}), which in turn gives a $\bigO(\log n)$ approximation for CC by the approximation preserving reduction from \cite{demaine2006correlation}.

We retrieve the same result by directly applying the CKR partitioning to $\lpccc$. Combining the following two lemmas gives an overall $\bigO(\log n)$ guarantee in expectation. 

\begin{lemma}
\label{len:ckrpos}
$\text{Pr}[\mathcal{P}(u) \neq \mathcal{P}(v)] \leq \bigO(\log n)\cdot x_{uv}$.
\end{lemma}
\begin{proof}
This result holds for arbitrary finite metric space. The probability that two points $u$ and $v$ at distance $x_{uv}$ apart are in different CKR clusters is known to be at most $\bigO(\log n)\cdot x_{uv}$, see for example \cite{calinescu2005approximation}.
\end{proof}

\begin{lemma}
$\text{Pr}[\mathcal{P}(u) = \mathcal{P}(v)] \leq 1-x_{uv}$.
\end{lemma}
\begin{proof}
Note that $B(u,r) \cap B(v,r) = 0$ when $x_{uv}>2r$.
Therefore, integrating the second term in Eq.~(\ref{eq:expcost}) gives
\begin{equation*}
\text{Pr}[\mathcal{P}(u) = \mathcal{P}(v)] = 2\int_{0}^{1/2}\text{Pr}[\mathcal{P}(u) = \mathcal{P}(v)|R=r] \,dr  \leq 2 \displaystyle \int_{x_{uv}/{2}}^{1/2} \,dr = 1-x_{uv}.
\end{equation*}
\end{proof}

\subsection{Complete signed graphs}
In order to reduce the $\bigO(\log n)$ factor in Lemma~\ref{len:ckrpos}, one cannot simply charge the expected cost of a positive edge to its own LP value.
We present an alternative charging scheme that gives a 3-approximation for complete signed graphs.

\paragraph{Charging scheme.}
The idea is to reallocate some of the terms in Eq.~(\ref{eq:expcost}) to other edges, as given by allocation functions $f_{uv}: [0,1/2) \rightarrow \mathbb{R}_{\ge 0}$, such that
\begin{equation}
\label{eq:reallocatecost}
\Exp[\text{cost}(\mathcal{P})| R=r] = \sum_{uv \in {V \choose 2}} f_{uv}(r),
\end{equation}
with $2\int_{0}^{1/2} f_{uv}(r) \,dr \leq 3\cdot\text{LP}(uv)$. Here, $\text{LP}(uv)$ denotes the LP cost of pair $uv$, which is $x_{uv}$ if $uv \in E^+$ and $1-x_{uv}$ if $uv \in E^-$.
To abbreviate the notation, let $B_u$ denote $B(u,r)$.
\begin{definition}
A positive edge $uv \in E^+$ is called \emph{short} if $x_{uv} \leq r$, otherwise it is called \emph{long}.
\end{definition}

Decompose the conditional cost in Eq.~(\ref{eq:expcost}) of a short positive edge $uv$ as
\begin{equation}
\frac{|B_u \triangle B_v|}{|B_u \cup B_v|} = \sum_{z\in B_u \setminus B_v}\frac{1}{|B_u \cup B_v|} + \sum_{z\in B_v \setminus B_u}\frac{1}{|B_u \cup B_v|}
\end{equation}
This cost is reallocated according to the following rules.
For every $z \in B_u \setminus B_v$, if $uz \in E^-$ then $uz$ receives $\frac{1}{|B_u \cup B_v|}$ amount of charge.
Instead, if $uz \in E^+$ then the edge $zv$ receives $\frac{1}{|B_u \cup B_v|}$ amount of charge (regardless of the sign of $zv$). In this setup, we call $uz$ a \emph{near} edge and $zv$ a \emph{far} edge associated with $uv$. Do the same for every $z \in B_v \setminus B_u$.

Lemmas~\ref{lem:allopos} and~\ref{lem:alloneg} give upper bounds on the allocation function $f$. If some intervals are not well-defined, then they are understood to be empty. Using Corollaries~\ref{cor1} and~\ref{cor2}, combined with a straightforward law of total expectation and linearity on Eq.~(\ref{eq:reallocatecost}), it follows that \textsc{CKR-Pivot} has an expected approximation ratio of 3.

\subsubsection{Positive edges}
\begin{lemma}
\label{lem:allopos}
If $uv \in E^+$, then
\begin{equation}
f_{uv}(r) \leq 
\begin{cases} 
1, & 0 \leq r < \frac{x_{uv}}{2}, \\
2, & \frac{x_{uv}}{2} \leq r < x_{uv}, \\
0, & x_{uv} \leq r <  \frac{1}{2}.
\end{cases}   
\end{equation}
\end{lemma}
\begin{proof}
The case $x_{uv} \leq r$ is immediate, because then $uv$ is a short edge and hence all its cost is reallocated to other edges, while also never receiving cost from other short edges.

In case $r<x_{uv}$, the edge $uv$ is a long edge and it receives two types of costs: $(i)$ its own cost, and $(ii)$ reallocated cost $\frac{1}{|B_z \cup B_u|} \leq \frac{1}{|B_u|}$ (resp. $\frac{1}{|B_z \cup B_v|} \leq \frac{1}{|B_v|}$) from a short positive edge $zu$ (resp. $zv$).
Note that there are at most $|B_u \cap B_v|$ edges of type $(ii)$, indeed, this reallocation only happens when $z \in B_u$ and $z \in B_v$. So edge $uv$ receives a cost of at most
\begin{align}
\label{eq:ineq1}
f_{uv}(r) \leq \frac{|B_u \triangle B_v|}{|B_u \cup B_v|} + \frac{|B_u \cap B_v|}{|B_u|} + \frac{|B_u \cap B_v|}{|B_v|} \leq 2.
\end{align}
The last inequality in Eq.~(\ref{eq:ineq1}) can be proven by using the definitions of the set operators, putting the terms on a common denominator and rearranging them. It remains to show the case $r < \frac{x_{uv}}{2}$. Here we use that $B_u \cap B_v = \emptyset$, and thus Eq.~(\ref{eq:ineq1}) gives $f_{uv} \leq 1$.
\end{proof}

\begin{corollary}
\label{cor1}
If $uv \in E^+$, then $2\int_{0}^{1/2} f_{uv}(r) \,dr \leq 3x_{uv}$.
\end{corollary}
\begin{proof}
Evaluate the integral using Lemma~\ref{lem:allopos}, by distinguishing two cases $x_{uv} \leq 1/2$ and $x_{uv} > 1/2$.
\end{proof}

\subsubsection{Negative edges}
\begin{lemma}
\label{lem:alloneg}
If $uv \in E^-$, then
\begin{equation}
f_{uv}(r) \leq 
\begin{cases} 
0, & 0 \leq r < \frac{x_{uv}}{2}, \\
3, & \frac{x_{uv}}{2} \leq r < x_{uv}, \\
2, & x_{uv} \leq r <  \frac{1}{2}.
\end{cases}   
\end{equation}
\end{lemma}

\begin{proof}
Let us start with the case $r < x_{uv}$. In this case, in addition to receiving its own cost, $uv$ can also receive cost by being the far edge of an adjacent short positive edge. Similarly to Eq.~(\ref{eq:ineq1}) we have
\begin{align}
\label{eq:ineq2}
f_{uv}(r) \leq \frac{|B_u \cap B_v|}{|B_u \cup B_v|} + \frac{|B_u \cap B_v|}{|B_u|} + \frac{|B_u \cap B_v|}{|B_v|} \leq 3,
\end{align}
where the last inequality follows simply because each term is at most one. Note also that in case $r < \frac{x_{uv}}{2}$, then $B_u \cap B_v = \emptyset$ and therefore $f_{uv}(r)=0$ by Eq.~(\ref{eq:ineq2}). Only the case $x_{uv} \leq r$ remains, in which $uv$ receives its own cost, while it might also receive cost as the near edge of an adjacent short positive edge. If the adjacent short edge is incident to $u$, then it is not hard to see that there are most $B_u \setminus B_v$ such edges. A similar observation holds for short edges incident to $v$.
So edge $uv$ receives a cost of at most
\begin{align}
\label{eq:ineq3}
f_{uv}(r) \leq \frac{|B_u \cap B_v|}{|B_u \cup B_v|} + \frac{|B_u \setminus B_v|}{|B_u|} + \frac{|B_u \setminus B_v|}{|B_v|} \leq 2.
\end{align}
The last inequality in Eq.~(\ref{eq:ineq3}) can be proven by putting all terms on a common denominator and rearranging them.
\end{proof}

\begin{corollary}
\label{cor2}
If $uv \in E^-$, then $2\int_{0}^{1/2} f_{uv}(r) \,dr \leq 3(1-x_{uv})$.
\end{corollary}
\begin{proof}
Evaluate the integral using Lemma~\ref{lem:alloneg}, by distinguishing two cases $x_{uv} \leq 1/2$ and $x_{uv} > 1/2$.
\end{proof}

\section{Discussion}
\label{sec:discussion}
We presented some negative and positive results on the constrained correlation clustering problem and variants.
Using a relatively simple reduction from vertex cover, we show that if the unique games conjecture is true, then it is NP-hard to obtain a $2-\epsilon$ approximation for the constrained correlation clustering, cluster deletion and minimum weakness strong triadic closure problems. This significantly improves the previous known hardness results on these problems, where often only APX-hardness or NP-hardness for approximation within some small specified constant close to one was known.

On the positive side, we propose a standard CKR decomposition as an approximation algorithm for (constrained) correlation clustering. Unfortunately, we are only able to prove a 3-approximation guarantee, which matches the guarantee of the pivot-based algorithm by \cite{van2009deterministic}.
Although, after extensive numerical testing, it appears the current analysis is not tight and therefore we conjecture the following.

\begin{conjecture}
\label{conj}
\textsc{CKR-Pivot} is a smaller-than-3 approximation for CCC.
\end{conjecture}

CKR decompositions can be a useful alternative to more traditional pivot-based approaches. One advantage is its flexibility, in the sense that it can be applied to several different variants of correlation clustering (one simply has to change the LP). For example, in \emph{multilayer} correlation clustering over $L$ layers \cite{miyauchi2026multilayer}, it only takes a few lines of calculation to show that \textsc{CKR-Pivot} has an approximation guarantee of $\bigO(L\log n)$ on general signed graphs. This matches the result of \cite{miyauchi2026multilayer} who showed a similar guarantee by using a region growing algorithm, but it required a significantly longer proof.

\paragraph{Acknowledgments}
The author used ChatGPT 5.6 Sol/Plus for finding the charging scheme in Section 3.3.

\bibliographystyle{alpha}
\bibliography{references}

\end{document}